\documentclass{article}
\usepackage{amsthm}

\newfont{\seaddfnt}{phvr8t at 8pt}

\usepackage{amsmath,amsfonts}
\usepackage[utf8]{inputenc}
\usepackage{url}

\usepackage{booktabs}

\makeatletter
\newif\if@restonecol
\makeatother

\usepackage[lined,ruled,vlined]{algorithm2e}

\newtheorem{thm}{Theorem}[section]
\newtheorem{lem}[thm]{Lemma}
\newtheorem{defn}[thm]{Definiton}
\newtheorem{prop}[thm]{Proposition}
\newtheorem{cor}[thm]{Corollary}
\newtheorem{conj}[thm]{Conjecture}
\newtheorem{exmp}[thm]{Example}

\newcommand{\GO}{\mathcal{ O}}
\newcommand{\sO}{\mathcal{ O}\tilde\ }
\newcommand{\Z}{\mathbb{Z}}

\newcommand{\minpoly}{{P_\text{min}^A}}
\newcommand{\charpoly}{{P_\text{char}^A}}
\urldef\jgdemail\url{Jean-Guillaume.Dumas@imag.fr}
\urldef\cpemail\url{Clement.Pernet@imag.fr}
\urldef\bdsemail\url{saunders@cis.udel.edu}

\begin{document}
\title{On finding multiplicities of characteristic polynomial factors
  of black-box matrices\thanks{Saunders supported by National Science
    Foundation Grants CCF-0515197, CCF-0830130.}.
}
\author{
Jean-Guillaume Dumas\thanks{Laboratoire J. Kuntzmann, Universit\'e de
  Grenoble. 51, rue des Math\'ematiques, umr CNRS 5224, bp 53X, F38041
  Grenoble, France, \jgdemail}
\and Cl\'ement Pernet\thanks{Laboratoire LIG, Universit\'e de
  Grenoble. umr CNRS, F38330 Montbonnot, France. \cpemail}
\and B. David Saunders\thanks{University of Delaware, Computer and
  Information Science Department. Newark / DE / 19716, USA. \bdsemail}
}

\maketitle

\begin{abstract}
We present algorithms and heuristics to compute the characteristic polynomial of a
matrix given its minimal polynomial. 
The matrix is represented as a black-box, i.e., by a function to compute
its matrix-vector product. 
The methods apply to matrices either over the integers or over a large enough finite field.
Experiments show that these methods perform efficiently in practice. Combined in
an adaptive strategy, these algorithms reach significant speedups in practice
for some integer matrices arising in an application from graph theory.
 \end{abstract}

\section{Introduction}
Computing the characteristic polynomial of an integer matrix is a
classical mathematical problem. 
It is closely related to the computation of the Frobenius normal form
which can be used to test two matrices for similarity, or computing invariant
subspaces under the action of the matrix.
Although the Frobenius normal form contains more information on 
the matrix than the characteristic polynomial, most algorithms to
compute it are based on computations of characteristic polynomials (see
for example \cite[\S 9.7]{Storjohann:2000:thesis}).

Several matrix representations are used in computational linear algebra. In the
dense representation, a $m\times n$ matrix is considered as the array of all the $m\times n$
coefficients. 
The sparse representation only considers non-zero coefficients using different possible
data structures. In the black-box representation,  the matrix is viewed  as a linear
operator, and no other operation than the application to a vector is
allowed. Though constraining, this limitation preserves the structure or
sparsity of the matrix and is therefore especially well suited for very large
sparse or structured matrices.

Computation of the characteristic polynomial of \textit{dense} matrices has already been well
studied both in theory and practice: over a finite field,
\cite{Keller-Gehrig:1985:charpoly,Pernet:2007:charp} set the best complexity
(using respectively a deterministic and a probabilistic algorithm), and 
\cite{jgd:2005:charp,Pernet:2007:charp} propose efficient implementations.
Over the integers, the best complexity is achieved in
\cite{Kaltofen:2005:CCDet}, but currently the  most efficient implementations 
rely on the Chinese remainder algorithm \cite{jgd:2005:charp}. 

In the latter article, a competitive approach is introduced that limits the use
of the Chinese remainder algorithm to the computation of the minimal polynomial.
The characteristic polynomial is then recovered by determining
the multiplicities of its irreducible factors. This task is done using the
deterministic algorithm for the characteristic polynomial over a randomly
chosen prime field.

In the black-box model, the minimal polynomial is used as a building block for
many algorithms over a finite field. Adapted from the iterative numerical
methods (Lanczos, Krylov), the Wiedemann minimal polynomial algorithm
\cite{Wiedemann:1986:SSLE,Kaltofen:1991:SSLS} has excellent asymptotic complexity
and is used in efficient black-box linear algebra software such as 
\texttt{LinBox}\footnote{\url{www.linalg.org}}.

However less is known concerning the characteristic
polynomial of black-box matrices. It is an open problem
\cite[Open Problem 3]{Kaltofen:1998:Open} to compute 
the characteristic polynomial as efficiently as the minimal polynomial, using
the Wiedemann method. The latter uses $\GO(n)$ products of a square $n\times n$
matrix by a vector and $\GO(n^2(\log n)^{\GO(1)})$ additional arithmetic
operations with $\GO(n)$ extra memory storage.
Eberly gives an algorithm  using $\GO(n)$ matrix vector products,
and $\GO(\phi n^2)$ additional operations, where $\phi$ is the number of
invariant factors of the matrix~\cite{Eberly:2000:BBFDOSF}. 
In the worst case, $\phi = \GO(n)$ and the
algorithm does not improve on the complexity of dense algorithms.
Villard proposed in \cite{Villard:2000:Frob} a black box
algorithm to compute the Frobenius normal form and therefore the characteristic
polynomial in $\GO(\sqrt{n}\log(n))$ computations of minimal polynomials 
and $\GO(n^{2.5} (\log n)^2 \log\log n)$ additional field operations. 

Instead, we propose here several algorithms and heuristics focusing on
efficiency in practice.
The general strategy is to compute the minimal
polynomial using  Wiedemann's algorithm and decompose it into irreducible factors.
There only remains to determine to which multiplicity each of
these factors appear in the characteristic polynomial. 
In section \ref{sec:multip} we propose
several methods to determine these multiplicities.
Adaptive combination of them is discussed in 
section \ref{sec:adap}. 
Under a conjectured hypothesis the latter is shown to
require  $\GO\ (n\sqrt{n})$ matrix vector products which improves by a logarithmic
factor on the complexity of Villard's algorithm.

Lastly, an algorithm for the computation over the ring of integers is derived in
section \ref{sec:spint}.
It is based on the  multifactor Hensel lifting of a gcd-free basis, following Storjohann~\cite{Storjohann:2000:Frob}. 
The benefit of this approach is verified by experiments
presented in section \ref{sec:exp}. Several sparse matrices are considered,
including a set of  adjacency matrices of strongly regular graphs, coming from an
application in graph theory.





\section{Three methods for computing multiplicities}\label{sec:multip}

In this section we consider a matrix $A$ over a finite field
$K=\text{GF}(q)$.
Let 
$\minpoly = \prod_{i=1}^k P_i^{e_i}
$
be the decomposition of the minimal polynomial of $A$ in irreducible monic
factors. The characteristic polynomial is then 
\begin{equation}\label{eq:charpol}
\charpoly = \prod_{i=1}^k P_i^{m_i}
\end{equation} 
for some $m_i \geq e_i$. We also denote by $d_i$ the degrees of
the factors: $d_i = \text{deg}(P_i)$.

To recover the multiplicities $m_i$, we will present three techniques,
based on black-box computations with the matrix $A$:
the nullity method (\S \ref{sec:nullity}) uses the
rank of $P_i(A)$ to reveal information on the multiplicity~$m_i$, 
the combinatorial search (\S \ref{sec:combs}) is a branch and bound
technique to solve the total degree equation whose integral unknowns are the
multiplicities and the index calculus technique (\S \ref{sec:dlog}) uses a
linear system solving based  on the discrete logarithm  
 of equation \ref{eq:charpol} evaluated in random values.

\subsection{The nullity method}
\label{sec:nullity}
\begin{defn}
The nullity  $\nu(A)$ of a matrix $A$  is the dimension of its nullspace.
\end{defn}
We also recall the following definitions:

The \textit{companion matrix} of the monic polynomial 
$P=X^d + \sum_{i=0}^{d-1} a_i X^i$ is the matrix
{\scriptsize
$
\renewcommand{\arraystretch}{.6}
\setlength{\arraycolsep}{.35\arraycolsep}
C_P = 
\left[\begin{array}{cccc}
0&      &     & -a_0\\
1&0     &     & -a_1\\
 &\ddots&\ddots&\vdots\\
 &      &  1   &-a_{d-1}
\end{array}\right].
$}
Its minimal polynomial and its characteristic polynomial are equal to $P$.

The \textit{block Jordan matrix} of an irreducible polynomial $P$ of degree $d$ to a
power $k$ is the $kd\times kd$ matrix $J_{P^k}$ of the form
$
J_{P^k}=
\renewcommand{\arraystretch}{.5}
\setlength{\arraycolsep}{.5\arraycolsep}
\left[\begin{array}{cccc}
  C_P&B&&\\
  &\ddots&\ddots&\\
  &&C_P&B\\
  &&&C_P
\end{array}\right]
$
where the $d\times d$ matrix $B$ is filled with zeros except for $B_{d,1}=1$.
Its minimal polynomial and its characteristic polynomial are equal to $P^k$.
This definition extends the usual notion of Jordan blocks for $d=1$.
 
The \textit{Frobenius normal form} of a Matrix $A$ is the unique  block diagonal matrix
$F=\text{Diag}(C_{f_0},C_{f_1},\dots)$ such that $A=U^{-1}FU$ for a nonsingular matrix
$U$. The polynomials $f_i$ are the invariant factors of $A$ and satisfy $f_0 =
\minpoly$ and $f_{i+1}$ divides $f_{i}$ for all $i\geq0$.

The \textit{primary form} of a Matrix $A$ (also called the second Frobenius form in
\cite{Gantmacher:1959:TMone}) is a further decomposition of the Frobenius
normal form where each companion block $C_{f_i}$ is replaced by a block
diagonal matrix  $\text{Diag}(J_{g_1^{k_1}}, J_{g_2^{k_2}},\dots)$. 
  The $g_j$ are the irreducible factors of $f_i$, with the respective
  multiplicities $k_j$. 
The primary form is unique up to the order of the blocks. 

%
\begin{exmp}
Consider the matrix in Frobenius normal form 
{\small $$A= \text{Diag}(C_{X^5-6X^4 +14X^3 - 16X^2 + 9X-2},C_{X^2-2X+1})$$} over
$\text{GF}(5)$. The corresponding 
 primary form is the matrix {\small $$B = \text{Diag}(J_{(X^2-2X-1)^2}, J_{X-2},
   J_{X^2-2X-1}).$$}\\
{\scriptsize
$
\renewcommand{\arraystretch}{.6}
\setlength{\arraycolsep}{.5\arraycolsep}
A= 
\left[
\begin{array}{ccccccc}
0& 0 & 0 & 0 & 2 &  &\\
1& 0 & 0 & 0 & -9 &  &\\
 & 1 & 0 & 0 & 16 &  &\\
 &   & 1 & 0 & -14 &  &\\
 &   &   & 1 & 6 &  & \\
 &   &   &   &   & 0& -1\\
 &   &   &   &   & 1& 2\\
\end{array}
\right]
,
B= 
\left[
\begin{array}{ccccccc}
0& 0  & 0 & -1  &  &  &\\
1& 0 & 0 & 4 &  &  &\\
 & 1 & 0 & -6 &  &  &\\
 &   & 1 & 4 &  &  &\\
 &   &   &   & 2&  & \\
 &   &   &   &  & 0& -1\\
 &   &   &   &  & 1 &2\\
\end{array}
\right]
$}
\end{exmp}
The method of the nullity is based on the following lemma:
\begin{lem}
Let A be a square matrix and 
let $P$ be an irreducible polynomial 
of degree $d$, of multiplicity $e$ in the minimal polynomial of $A$,
and of multiplicity $m$ in the characteristic polynomial of $A$.  
Then $\nu(P^{e}(A))=md.$
\end{lem}
\begin{proof}
Let F be the primary form of $A$ over $K$: $F=U^{-1}AU$ for a non singular
matrix $U$.
$F$ is block diagonal of the form 
$\text{Diag}(J_{P_j^{e_{j}}})$.
Then $P^{e}(A) = U^{-1}P^{e}(F)U =
U^{-1}\text{Diag}(P^{e}(J_{P_j^k}))U$.
On one hand $P^{e}$ annihilates the blocks $J_{P_j^k}$ where $P = P_j$ and
$k\leq e$.
On the other hand, the rank of $P^{e}(J_{P_j^k})$ is full for $P \neq P_j$,
since $P$ and $P_j$ are relatively prime.
Thus the nullity of $P^{e}(A)$ exactly corresponds to the
total dimension of the blocks $J_{P_j^k}$ where $P = P_j$, which is $md$.
\end{proof}
From this lemma the following algorithm, computing the multiplicity $m_i$ of an
irreducible factor $P_i$ is straight-forward:
%
%
\begin{algorithm}
  \dontprintsemicolon
  \caption{\texttt{Nullity}}\label{alg:nullity}
  \KwData{$A$: an $n\times n$ matrix over a Field $K$,}
  \KwData{$P$: an irreducible factor of $\minpoly$,}
  \KwData{$e$: the multiplicity of $P$ in $\minpoly$,}
  \KwResult{$m$: the multiplicity of $P$ in $\charpoly$.}
  \Begin{
      $r = \text{rank}(P^{e}(A))$\;
      \Return $m = (n-r)/\text{degree}(P)$\;
    }
  \end{algorithm}

\begin{prop}
Algorithm \ref{alg:nullity} computes the multiplicity of $P$ in
the characteristic polynomial of an $n \times n$ matrix $A$ using $\GO(edn\Omega)$ 
field operations, where $\Omega$ is the cost of a matrix-vector
product with $A$, $d$ is the degree of $P$ and $e$ is its multiplicity
in the minimal polynomial.
\end{prop}
\begin{proof}
Using Horner's rule, the matrix $P(A)$ can be written as $a_0I_n
+A(a_1I_n+A(a_2I_n+\dots))$. Hence, applying a vector to this blackbox
only requires $d$ applications of a vector to the blackbox $A$,
ie. $\GO(d\Omega)$ field operations.
Thus applying a vector to $P^{e}(A)$ costs $\GO(ed\Omega)$.
Lastly, the rank of this matrix can be
computed in $\GO(edn\Omega)$ field operations, using 
Wiedemann's algorithm combined with preconditioners~\cite{jgd:2002:villard}.
\end{proof}
This algorithm is therefore suitable for irreducible factors $P$ where the
product $e d$ is small.

Now if $e$ is large, the computation of $\text{rank}(P^{e}(A))$ may be too
 expensive. Still, some partial knowledge on the multiplicity can be
recovered from the rank of the first powers of $P(A)$. This can help to shorten
the computation of other algorithms, as will be shown in section \ref{sec:adap}.
We now describe how these partial multiplicities can be recovered.

The multiplicity $m$ of $P$ is formed by the contribution of several blocks of the 
type $J_{P^j}$ for $j\in [1\dots e]$ in the primary form of $A$.
Whereas the blocks with small $j$ can be numerous, there must be few  blocks
with large $j$, due to the limitation of the total dimension (since $e$ is
large). 

We denote by $n_{i,j}$ the number of occurrences of $J_{P_i^j}$ in the
primary form of $A$. From the determination of the $n_{i,j}$, we can
directly deduce the multiplicity $m_i$ by the relation 
\begin{equation}\label{eq:mi} m_i =
\sum_{j=1}^{e}{j n_{i,j}}.
\end{equation}
We now show how to compute the $n_{i,j}$ for
small $j$, using algorithm \ref{alg:nullity}. 

\begin{lem}
Let $P$ be an irreducible polynomial of degree $d$ over a finite field $K$
and $k$ and $e\geq 1$ be two integers. Then 
$\nu(P^k(J_{P^e})) = \text{min}(k,e)d \\ 
$
\end{lem}
\begin{proof}
Let $A=J_{P^e}$ and $B=P^k(A)$. 
If $k\geq e$, then $P^k$ is a multiple of the minimal polynomial of
$A$. Thus $B$ is the zero matrix, and its nullity equals its dimension: $ed$.


Now suppose $k<e$.
Let $\overline{K}$ be an extension of $K$ such that $P$ splits into
$d$ degree one factors $P_i$ over $\overline{K}$. Since any finite field is a perfect field, these factors are distinct.

The minimal polynomial of $A$ over $\overline{K}$ is still $P^e$. 
Consequently, the Frobenius normal form of A over
$\overline{K}$ is $C_{P^e}$ and  
its primary form is $F=\text{Diag}(J_{P_i^e})$. More precisely,
there exists $U \in M_n(\overline{K})$ such that
$A = U^{-1}FU$.
We have therefore $B=U^{-1}P^k(F)U = U^{-1}\text{Diag}(P^k(J_{P_i^e}))U.$

First consider the case $k=1$: the minimal polynomial of each
$P_i(J_{P_i^e})$ is $X^e$ and so is the minimal polynomial of each
$P(J_{P_i^e})$ (since the $P_i$ are relatively prime). Hence the primary form
of $P(J_{P_i^e})$ is $J_{X^e}$. Therefore there
exist $V\in M_n(\overline{K})$ such that 
$$
B=U^{-1}V^{-1}
Diag(\underbrace{J_{X^e},\dots,J_{X^e}}_{d \text{ times}}) V U.
$$

Lastly the rank of $J_{X^e}$ being $e-1$, we deduce that
$\text{rank}(B)= d(e-1)$. The nullity of $B$ is therefore $\nu(B)=d$.

For the general case, we have
$$
B =
U^{-1}V^{-1}\text{Diag}(\underbrace{(J_{X^e})^k,\dots,(J_{X^e})^k}_{d
\text{ times}})VU.
$$
Now $J_{X^e}$ is $e\times e$ and nilpotent with ones on the super-diagonal
so that its $k$-th power has rank $\max(0,e-k)$.
Thus, $\text{rank}(B)=\max(0,e-k)d$ and $\nu(B)=\min(e,k)d$.
\end{proof}
We now apply this result to the irreducible factors of the minimal polynomial
and denote the nullity of $P_i^j(A)$ by 
$\nu_{i,j}=\nu(P_i^j(A))$. 

First, the nullity of $P_i(A)$,  can be decomposed
into the sum of the nullities of each $P_i(J_{P_i^k})$ for every $k\leq e_i$:
\begin{equation}
\nu_{i,1} = \sum_{k=1}^{e_i}{n_{i,k}d_i}\\ \label{eq:p1}
\end{equation}
Now applying $P_i^j$ to $A$, every $P_i^j(J_{P_i^k})$ for $k\leq j$
will be a zero matrix and therefore contribute with $kd_i$ to the
nullity. Otherwise, if $k > j$, the contribution to the nullity
remains $jd_i$. Therefore we have:
\begin{equation}
\nu_{i,j} = \sum_{k=1}^{j}{n_{i,k}kd_i} +  \sum_{k=j+1}^{e_i}{n_{i,k}jd_i}\\ \label{eq:pj}
\end{equation}
From these two equations, we deduce the $n_{i,j}$: first we have
$$
\frac{1}{j-1}\nu_{i,j-1} =  \frac{1}{j-1} \sum_{k=1}^{j-1} n_{i,k}kd_i +
n_{i,j}d_i + \sum_{k=j+1}^{e_i} n_{i,k} d_i.
$$
Now, since:
$  \sum_{k=j+1}^{e_i} n_{i,k}d_i = \nu_{i,j+1}-\nu_{i,j},$
the number of occurrences directly is:
$$
n_{i,j} = \frac{1}{d_i}\left(\frac{1}{j-1}\nu_{i,j-1}  + \nu_{i,j} -
\nu_{i,j+1}\right) - \frac{1}{j-1} \sum_{k=1}^{j-1} n_{i,k}k.
$$ 
Therefore we obtain corollary \ref{cor:occ} giving the expression of the $n_{i,j}$:
\begin{cor}\label{cor:occ}
\begin{eqnarray*}
n_{i,1} &=& (2\nu_{i,1}-\nu_{i,2})/d_i \\
n_{i,j} &=& \frac{1}{d_i}\left(\frac{1}{j-1}\nu_{i,j-1}  + \nu_{i,j} -
\nu_{i,j+1}\right) \\
&&- \frac{1}{j-1} \sum_{k=1}^{j-1} n_{i,k}k \ \ \forall j \in
[1\dots e_i]\\
n_{i,e_i}&=& \frac{\nu_{i,e_i}}{e_id_i} - \frac{1}{e_i}\sum_{k=1}^{e_i-1}n_{i,k}k
\end{eqnarray*}
\end{cor}
The last formula for $n_{i,e_i}$ is given for the sake of completeness: 
in practice, one will never compute every $n_{i,j}$, since 
one would rather directly compute  the nullity of  $P_i^{e_i}(A)$ instead, to
deduce the multiplicity $m_i$ from algorithm~\ref{alg:nullity}.

\subsection{The combinatorial search}
\label{sec:combs}


In the following, we want to determine the values  of the unknown
$n_{i,j}$.
They must satisfy the total degree
equation:
\begin{equation}\label{eq:degree}
n = \sum_id_i\sum_{j=1}^{k_i}jn_{i,j}.
\end{equation}

We can also discriminate potential candidates using the trace:
the degree $n-1$ coefficient of
the characteristic polynomial is the negative of  
the trace of the matrix. Denote by $t_i$ the degree $n-1$ coefficient
of an irreducible factor $P_i =
X^{d_i}+t_iX^{d_i-1}+\ldots$. Then the degree $n-1$ coefficient of $\prod_i
P_i^{m_i}$ is $ \sum_i t_i m_i $. We thus  have the trace test:
\begin{equation}\label{eq:tr}
\text{Tr}(A) = -\sum_i t_i m_i = -\sum_i t_i\sum_{j=1}^{k_i}jn_{i,j}.
\end{equation} 
\newcommand{\branchandcut}{\texttt{Branch-and-Cut}}
In a pure black-box model, the trace can be computed using $n$ matrix-vector
products.  For many sparse or structured matrix representations, a faster method is available as well.

Then it suffices to use e.g. a {\branchandcut} algorithm to compute
all the integer $k$-tuples satisfying both equations (\ref{eq:degree})
and (\ref{eq:tr}).
Of course, if some of the unknowns $n_{i,j}$ are already computed
(e.g. by the nullity method) the set of candidates is accordingly reduced.

The remaining candidates will then be
discriminated by  evaluations of the characteristic polynomial
at random values, i.e. computations of determinants of $\lambda
I-A$ matrices. Indeed, we have efficient
methods of computing the determinant of a black-box matrix (see
e.g. \cite[\S 3.1 Determinant Preserving
Preconditioners]{Turner:2002:these} and references therein).
Algorithm \ref{alg:large:factor} sums up this combinatorial search
strategy.
\newcommand{\completelargefactors}{\texttt{Combinatorial-search}}
\begin{algorithm}
  \dontprintsemicolon
  \caption{\completelargefactors}\label{alg:large:factor}
  \KwData{$A$, an $n\times n$ matrix}
  \KwData{$D=(d_i)_i$, the degrees of the irreducible factors $P_i$ of $\charpoly$}
  \KwData{$M$, a set of precomputed $n_{i,j}$}
  \KwResult{$N=(n_{i,j})$ }
  \Begin{
       \tcc{using degree and trace constraints}
       sol = \branchandcut (A, D, M)\;
      \While{$\# sol > 1$}{
        Pick $\lambda \in K$ at random\;
        $\delta=\text{det}(\lambda I -A)$\;
        Discard any $N \in \text{sol}$ s.t. $\prod (P_i^j)^{n_{i,j}}(\lambda) \neq \delta$\;
      }
      \Return{$N = sol[1]$}
  }
\end{algorithm}

\subsection{Index calculus method}
\label{sec:dlog}

Evaluating equation (\ref{eq:charpol}) at a point $\lambda$ leads to an
equation over the finite field, where the  multiplicities $m_i$ are
the unknowns.
Inspired by index calculus techniques
\cite{Coppersmith:1986:IndexCalculus},
 the idea here is to consider the discrete
logarithm of such an equation (with an arbitrary choice of generator), to
produce a linear equation in the $m_i$. Taking several of these
equations for different
$\lambda_i$ forms a linear system of equations, with dimension $k$,
the number of
unknown multiplicities.

In this discussion the base field is $GF(q)$ and $q$ is sufficiently
large with respect to $n$ as discussed below. 
The characteristic polynomial evaluated at a given value $\lambda$  presents this 
equation in the unknown exponents $m_j$:
\begin{equation}
\prod_{j=1}^k  P_j^{m_j}(\lambda) = \det(\lambda I - A).
\end{equation}

Now, if $\lambda$ is not a root of the characteristic polynomial, taking the
discrete logarithm of these terms for a generator $g$ of the field leads 
to this equation modulo $q-1$:
\begin{equation}\label{eq:log}
\sum_{j=1}^k  m_j\log_g(P_j(\lambda)) \equiv
\log_g(\det(\lambda I - A)) \mod  q-1,
\end{equation}
which is linear in the unknowns $m_j$.
We can therefore build a $l\times k$ linear system by randomly choosing $l$ values
$\lambda_i$. This system is consistent since the multiplicities $m_i$
are a solution vector of this system. If the solution is unique, then it is the
vector of multiplicities over $\Z$.

The computation of this vector can either be done by a dense Gaussian
elimination over the ring $\Z_{q-1}$ or over a finite field $\Z_p$ 
where $p$ is a large prime factor of $q-1$ (larger than $n$).
In this last case, the result will be correct as long as the system
remains nonsingular modulo $p$.



Algorithm \ref{alg:dlog} describes this techniques in more details.

\newcommand{\dlogsys}{\texttt{Index-calculus}}
\begin{algorithm}
  \dontprintsemicolon
  \caption{\dlogsys}\label{alg:dlog}
  \KwData{A, an $n\times n$ matrix over a finite field $K=\text{GF}(q)$}
  \KwData{$P_i$, the irreducible factors of $\minpoly$,}
  \KwData{$S$, the set of indices of the unknown multiplicities,}
  \KwData{$Q = \prod_{j \notin S}P_j^{m_j}$, the partial product of the
    irreducible factors with known multiplicity $m_j$.}
  \KwResult{$\charpoly$  or ``fail"}
\Begin{
    $k=\#S$;  $l=0$; $H=[\ ]$;\; 
    Choose a generator $g$ of $K$\;
    Let $p$ be a prime factor of $q$\;
     \While{$rank(H) < k$}{
       $l=l+1$; \lIf{$l > n$}{\Return{``fail"}}\;
       Choose randomly $\lambda_l \in K$\; 
       \Repeat{ $\gamma \neq 0$ and $\alpha_j\neq 0$, $\forall j\in S$}{
         $\alpha_{l,j}=P_j(\lambda_l)$ for all $j$\;
         $ \gamma_l=Q(\lambda_l)$\;
       }
        \tcc{$H$ is extended to the size $l \times k$}
        Stack the row $[\log_g \alpha_{l,j} \mod (q-1)] \mod p$ to $H$\;
     }
    {$\mathcal K$} = \{ indices of the first $k$ independent rows
    of $H$\}\;
    Set $B$ = a $k\times k$ nonsingular matrix of these rows\;
    Compute $b=[\log_g(\det(\lambda_iI-A)) - \log_g(\gamma_i)]_{i \in {\mathcal K}}$\;
    Solve $Bx=b$\;
    \Return{$P=\prod_{j\in S}P_j^{x_j} Q$}
}
\end{algorithm}

Let $k$ be the number of unknown multiplicities. 
The first step is to find $k$ values $\lambda_i$ forming a non
singular system. Therefore, we propose, in algorithm \ref{alg:dlog},
to evaluate the system at more than $k$ points. The complexity of
forming a row of $G$ costs only $3\sum_{i=1}^k d_i$ arithmetic operations
using Horner's method. Therefore trying as many as $n$ different values for $\lambda$ is
a negligible cost. 
Furthermore, one could use fast multi-point evaluation to get
blocks of rows simultaneously (up to $n$ rows at a cost essentially linear in $n$).

The rank of $H$ is computed all along the process, each new row being
incrementally added to the triangular decomposition of the current matrix.
This Gaussian elimination (performed by the LQUP algorithm
\cite{Ibarra:1982:LSP} for example) 
also provides the indices of the first $k$ linearly independent
rows, and therefore the indices of the convenient $\lambda_i$.
Lastly the vector $b$ is formed, using only $k$ determinant computations.


In practice, it appeared, as in index calculus
\cite{Bender:1998:rdlog,Hess:2007:dlog}, that the
number of rows
required to get a
full rank matrix $B$ is always quite close to $k$.
However, we do not have a proof of this property, and we therefore 
state it as the conjecture
\ref{conj:dlog}. 
\begin{conj}\label{conj:dlog}
Let $A$ be a $n\times n$ matrix over $\mathbb{Z}$ and $P_1\dots P_k$ be the
irreducible factors of its minimal polynomial.
Let $p>n$ be a prime chosen randomly in finite set. Let $q=1+\lambda p$ of the
form $r^k$ where $r$ is a prime number. Let $g$ be a generator of $\text{GF}(q)$
  and $(\lambda_1,\dots,\lambda_n)$  uniformly chosen at random in
  $\text{GF}(q)$.
Let $H=[h_{i,j}]$ where $h_{i,j}=(\log_g(P_j(\lambda_i)) \mod (q-1))\mod p$. 
Then $\text{rank}(H)=k$ with high probability.

\end{conj}

Informally, in our system the
evaluations at the $\lambda_i$ are independent and can be
considered as seeds for the pseudo-random generator of taking the
discrete logarithm of the polynomial evaluation. 
Therefore the entries of the system modulo $q-1$ 
are at least close to random entries as soon as the polynomials are
distinct.
Would they be true random values,  the singularity/nonsingularity of the matrices
would follow the analysis of e.g. \cite[Corollary
2.4]{Blomer:1997:rank}: 
if $L$ is a square matrix of uniformly random 
entries modulo $q-1$ and $p$ is a prime  diving $q-1$, then the
probability that $L \mod p$ is singular is of order $\frac{1}{p}$.

\begin{thm}Assuming conj. \ref{conj:dlog},
algorithm \ref{alg:dlog} is correct and its asymptotic complexity is 
$\GO( k n \Omega )$ where $\Omega\geq n$ is the cost of a multiplication of
$A$ by a vector.
\end{thm}
\begin{proof} Let $l$ be the
  number of rows required to get an invertible system, $l \geq k$.
Each determinant computation requires $\GO(n)$ application of $A$ to a vector
\cite[\S 3.1]{Turner:2002:these}. Building each row of the matrix requires a
Horner like evaluation of a polynomial with total degree less than $n$,
it therefore costs $\GO( n )$ operations. Triangularization of $G$
requires $\GO{(l k^{\omega-1})}$ operations. Solving the system $Bx=b$, knowing
the triangular decomposition of $B$ requires $\GO( k^2 )$
operations. The discrete logarithms can be tabulated
\cite{jgd:2004:dotprod} with $\GO(q)$ memory
(or to avoid this extra memory, one can compute the whole
sequence of powers of a generator of $K$, sort the matrix and vector
entries and find the correspondences with some $\sO(lk+q)$
extra field operations). 
The overall complexity is thus 
$\GO( k n \Omega + ln + lk^{\omega-1})$ which is $O(k n \Omega)$ 
when $l =\GO(n\Omega k^{2-\omega})$. 
\end{proof}
The algorithm stops arbitrarily when $l = n + 1$.  We see here that a larger $l$
is acceptable for the complexity result, but in our experiments a very small $l$
(e.g. $l \approx k$) always suffices. 
In the following sections, this algorithm will be used with
$k<\sqrt{n}$, thus giving an expected $\GO(n^{1.5}\Omega)$ complexity.

\section{Adaptive black-box algorithm over a finite field}\label{sec:adap}

We show in the present section how to combine the ideas of the
previous section together with already existing techniques to form an adaptive
algorithm computing the characteristic polynomial of a black-box matrix
over a finite field. The algorithm is adaptive in the sense of
\cite{jgd:2006:AHA}, meaning that it chooses the best variant depending
on discovered properties of its input.

We first combine the nullity method with the combinatorial search.
We then show  an algorithm improving on the asymptotic
complexity. Finally we give some improvements which are efficient in practice 
on typical matrices.

\subsection{Nullity method and combinatorial search}

These two algorithms are complementary: the nullity is efficient for the
determination of the multiplicites of factors of small degree, whereas the
combinatorial search is adapted to the large degree factors.

More precisely, algorithm \ref{alg:nullcombs} sorts the list of the unknown occurences
$n_{i,j}$ according to the increasing $jd_i$. The nullities are then computed
until there remain fewer than a fixed number $T$ of unknowns to be determined by 
combinatorial search.

\newcommand{\nullcomb}{\texttt{Nullity-comb-search}}

\begin{algorithm}
  \dontprintsemicolon
  \caption{\nullcomb}
  \label{alg:nullcombs}
      \KwData{A: an $n\times n$ matrix over a finite field,} 
      \KwData{$P_i$: the irreducible factors of $\minpoly$,}
      \KwData{$T$: a static threshold}
      \KwResult{$m_i$: the multiplicities of each $P_i$ in $\charpoly$.}
     \Begin{
         $\mathcal{E} = \{ (i,j) / i = 1\dots k, j = 1 \dots e_i\}$\;
         Sort $\mathcal{E}$ according to the increasing values of $jd_i$\;
         \While{($\# \mathcal{E} > T$)}{
           Pop $(i,j)$ from $\mathcal{E}$\;
           Compute $\nu_{i,j} = n- \text{rank}(P_i^j(A))$\;
         }
         \For{$i =  1 \dots k$}{
           Let $j_i$ be the largest index s.t.  $\nu_{i,j_i}$ is computed  \;
           \If{$j_i<e_i$}{
             Compute $\nu_{i,j_i+1}= n- \text{rank}(P_i^{j_i+1}(A))$\;
             }
           \lFor{$k =1 \dots j_i$}{Compute $n_{i,k}$  using cor. \ref{cor:occ}}\;
         }
         \completelargefactors$(A,(d_1,\dots,d_k),(j_1,\dots,j_k))$ \;
         \lFor{$i=1\dots k$}{$ m_i = \sum_{j=1}^{e_i}{jn_{i,j}}$\;}
         \Return{$m=(m_1,\dots,m_k)$}
       }
\end{algorithm}
The combinatorial search has exponential complexity. The threshold
$T$ must be small. In experiments, we found that $T=5$ was the best choice for
various matrices.
In the case of numerous factors with large degree, this approach is of reduced effeciveness. We propose in the next section how to combine it with a third algorithm.

\subsection{Nullity method and system resolution} \label{sec:improv}

The index calculus method also enables the design of a hybrid
algorithm. If the multiplicities of some factors have already been computed by
another method, we can limit the system to the unknown multiplicities only, thus
reducing its dimension.  

Suppose that the multiplicities $m_i$ of the factors
$P_i$ for $i \in \mathcal{C}$ are already known, then equation (\ref{eq:log}) reduces to 
{\small
\begin{equation}\label{eq:log2}
  \begin{split}
    \sum_{j \notin \mathcal{C}}  m_j\log(P_j(\lambda_i)) \equiv&
    \log(\det(\lambda_i I - A))\\
    &-\sum_{j\in \mathcal{C}} m_j\log(P_j(\lambda_i))\mod q-1.
  \end{split}
\end{equation}
}

A first simple hybrid approach is the following: the method of
the nullity  (section \ref{sec:nullity}) is applied to every degree one factor
with multiplicity one in the minimal polynomial, and the remaining factors 
are left to the index calculus method.

This approach is always worthy since the computation of the rank of $P_j(A)$,
for a degree one polynomial $P_i$ is cheaper than the computation of
$det(\lambda_iI-A)$ \cite{jgd:2002:villard}.


A second hybrid approach also introduces a combinatorial search to this
algorithm: the nullity method still handles the $t$ degree one factors as
previously. 
The remaining factors $P_i$ are sorted by decreasing degree. For a
convenient choice of $s$, a list of every possible assignment for the
multiplicities of the first $s$ factors is determined, using a combinatorial search.
Then for each partial assignment, the multiplicities of the remaining factors
are determined by the resolution of an index calculus system of the form:
{\small \begin{equation}
\begin{split}
 \sum_{j=s+1}^{k-t}  m_j\log(P_j(\lambda_i)) \equiv
\log_g(\det(\lambda_i I - A)) \\
- \log_g \left(\prod_{j=1}^s P_j^{m_j}(\lambda_i)\right) \mod q-1 ~\forall i .
\end{split}
\end{equation}
}

For each partial assignment, the system resolutions share the same matrix
$B$. Therefore the expensive part of it, namely the Gaussian elimination, can be
performed only once  at cost $\GO(m^3)$, where $m=k-s-t$.
There only remains to solve two triangular systems (in $\GO(m^2)$)
for each possible assignment. 
Lastly, the assignments will be discriminated against each other by a test on
the total degree. 
To sum up, this techniques makes it possible to balance the cost of the
computations of determinants, and the cost of the system solving, by
reducing the dimension of the system, but increasing the dimension of its right
hand side.
The most appropriate value for $s$ has to be determined dynamically, according
to the number of possible assignments induced, and using an estimate of the cost
function of this algorithm: e.g.
$$2mn\Omega+\frac{2}{3}m^3+4m^2 \tau_s$$
where $\tau_s$ denote the number of possible assignments for a chosen subset of
$s$ factors.





\subsection{Index calculus and k$^{th}$ invariant}

The best known black-box algorithm to compute the Frobenius normal form over a
field is given by Villard in \cite{Villard:2000:Frob}. It is proved that computing the
$k$th invariant factor of a matrix reduces to the computation of a minimal
polynomial of the input matrix with a rank $k$ additive perturbation.
Using a binary search technique, the algorithm only performs $\mu\text{log}(n)$
such computations, where $\mu$ is the number of distinct invariant factors of the matrix.
Since $\mu$ is smaller than $3\sqrt{n}/2$ and  an invariant factor can
be recovered using $\GO(n)$ matrix vector products, this corresponds
to a total number of $\GO(n^{3/2}\text{log}(n))$ matrix-vector products and an
additional cost of $\GO(n^{5/2}\text{log}^2(n)\text{loglog}(n))$ arithmetic
operations.

We propose in algorithm \ref{alg:spcharpoly1} an alternative approach combining
the index calculus method with computations of individual
invariant factors. 

\begin{algorithm}
  \dontprintsemicolon
  \caption{\texttt{black-box-charpoly} \label{alg:spcharpoly1}}
  \KwData{A: an $n\times n$ matrix over a finite field K}
  \KwResult{$\charpoly$  or ``fail"}
  \Begin{
      $f_1 = \texttt{InvFact}(1)$ \;
      Factor $f_1 = \prod_{i=1}^k P_i^{e_i}$ using Cantor-Zassenhaus\;
      Set $S=\{P_1,\dots,P_k\}$ and $j=2$\;
      \While{$(\#S > \sqrt{n})$}{
        $f_j = \texttt{InvFact}(j)$ \;
        \ForAll{$P_i \in S$}{
          Compute $\alpha$ s.t. $\text{gcd}(P_i^{e_i},f_j)=P_i^\alpha$\;
          \lIf{$\alpha = 0$} {$S = S \backslash \{P_i\}$\;}
              {$m_i \text{ += } \alpha$\;}
        }
      }
      \Return{\dlogsys $(A, (P_i), S, \prod_{j \notin S}P_j^{m_j})$}
    }
\end{algorithm}

The idea is to reduce the dimension of the index calculus system to $\sqrt{n}$ 
by computing a few of the first invariant factors of the matrix.

After each computation of an invariant factor $\Phi$, the multiplicity of
each irreducible polynomial $P_i$ is updated, and those $P_i$ that are no longer
in $\Phi$ are removed from the list of the factors with unknown multiplicity.

The $\texttt{while}$ loop is executed at most $\sqrt{n}$ times. Otherwise, there
would be more than $\sqrt{n}$ invariant factors having more than $\sqrt{n}$
irreducible factors, and the total degree would be larger than $n$.

Now the condition of exit for this loop ensures that the order of the
linear system will be smaller than $\sqrt{n}$. Therefore only $\sqrt{n}$
determinants will be computed and the overall number of blackbox matrix-vector products 
is $\GO(n\sqrt{n})$

The remaining multiplicities are then determined by the
index calculus method  described previously,
requiring at most $\GO(n^{3/2} \Omega) $ applications of the matrix to a vector.

Under the conditions of validity for the index calculus algorithm, this
heuristic improves on the computation time of Villard's algorithm by a logarithmic factor.

\section{Lifting over the integers}\label{sec:spint}

Storjohann gives in \cite{Storjohann:2000:Frob} a method for the
computation of the Frobenius normal form of a black-box integer
matrix. It is based on a computation of the minimal polynomial over
$\Z$ and on a computation of the Frobenius normal form over a prime
field. Then a gcd-free basis for the invariant factors over $\Z_p[X]$
is computed and lifted over $\Z[X]$. 

We use the same idea but just for the characteristic polynomial, and
not for all the invariant factors. It is
thus simpler since we don't need to ensure that the Frobenius form
of $A$ modulo $p$ equals the integer Frobenius form reduced modulo
$p$. We just need ensure that the minimal and characteristic 
polynomials of $A$ modulo
$p$ equal the minimal and characteristic polynomials over the integers
reduced modulo $p$.

The goal of the following algorithm \ref{alg:intcharp}
is to compute the
integer characteristic polynomial from the integer minimal polynomial
and the characteristic polynomial modulo some prime $p$, obtained via
the previous sections. Algorithm \ref{alg:intcharp} is just a
simplification of that of \cite{Storjohann:2000:Frob}:

\begin{algorithm}
  \dontprintsemicolon
  \caption{\texttt{Gcd-free lifting of the characteristic
      polynomial}\label{alg:intcharp}}
  \KwData{A: an $n\times n$ integer matrix, $\minpoly$ its minimal
    polynomial over $\Z$,}
  \KwData{$p$: a prime number,}
  \KwData{$\overline{\charpoly}$ the characteristic polynomial of $A \mod p$,}
  \KwResult{$\charpoly$ the characteristic polynomial of $A$ over $\mathbb{Z}$}
  \Begin{
      Compute $S$ the squarefree part of $\minpoly$\;
      $\overline{S} = S\mod p$,       $\overline{\minpoly} = \minpoly \mod p$\;

      Compute a modular gcd-free basis $(\bar{g}_1,\ldots,\bar{g}_\chi)$
      of $(\overline{S}, \overline{\minpoly}, \overline{\charpoly})$, together with exponents $(\mu_1,\ldots,\mu_\chi)$
      such that $\overline{\charpoly} = \prod \bar{g}_i^{\mu_i}$\;
      Apply Hensel lifting on the basis $(\bar{g}_1,\ldots,\bar{g}_\chi)$ to produce
      $(g_1,\ldots,g_\chi)$ so that $S \equiv g_1 \ldots g_\chi \mod p^k$
      and $(g_1,\ldots,g_\chi) \equiv (\bar{g}_1,\ldots,\bar{g}_\chi) \mod p$\;
      \Return{$\charpoly = \prod g_i^{\mu_i}$}
    }
\end{algorithm}
The integer minimal polynomial is computed via \cite[Theorem
3.3]{jgd:2001:jsc} with an $\GO(sd\Omega)$ probabilistic complexity,
where $s$ is the size of its integer coefficients and $d$, its
degree.
The characteristic polynomial modulo $p$ is computed via algorithm
\ref{alg:spcharpoly1} with an $\GO(n^{1.5}\Omega)$ complexity and
$\sO(n^{2.5})$ extra field operations.
Then the squarefree part \cite{Gerhard:2001:sqrfree} and 
the Hensel lifting of the gcd-free basis takes $\sO( n k )$
word operations with fast integer and  polynomial arithmetic
\cite[Theorem 15.18]{VonzurGathen:1999:MCA}.

The size of the coefficients of the integer minimal polynomial is bounded in the worst case by
$s \leq \frac{n}{2}(\log_2(n)+\log_2(||A||^2)+0.212)$ where $||A||$ is the
largest entry in absolute value of the matrix $A$, see \cite[Lemma 2.1]{jgd:2007:jipam}.
When the matrix entries are of constant size, $s=\sO(n)$  and
as the
degree of the minimal polynomial is bounded by $n$, the dominant
asymptotic cost is that of the integer minimal polynomial computation. 
This result is already in \cite{Storjohann:2000:Frob}.
In practice, however, the coefficients of the minimal
polynomial are often much smaller than the bound and than those of the
integer characteristic polynomial. Furthermore, the degree of the
minimal polynomial can be extremely small, especially for structured
or sparse matrices (see e.g. homology matrices
in \cite{jgd:2001:jsc}). In those cases, the dominant cost will 
be the computation of the characteristic polynomial modulo $p$.
Then, our algorithm enables faster computations since a factor of
$\sqrt{n}$ has been gained, as illustrated in table \ref{table:timings}.

A supplemental constraint can be introduced by 
computing det( $\lambda I - A$),
at random integer $\lambda$. 
Using e.g. \cite[Theorem 4.2]{jgd:2006:det} with \cite{Eberly:2006:sparse},
$\GO(\sqrt{n})$ of these 
can be done to speed-up the modular adaptive search,
without increased complexity.

\section{Experimental comparisons and applications}\label{sec:exp}

We have implemented some of the algorithms presented in the previous sections using the
\texttt{LinBox}\footnote{\url{www.linalg.org}} library for the black-box
computation of minimal polynomials, ranks and determinants.
In a first approach, we replaced the computation of the gcd-free basis of
algorithm \ref{alg:intcharp} by a factorization into irreducible factors, using
Hensel lifting. This algorithm is more expensive in the worst case, but the
efficient implementation by \texttt{NTL}\footnote{\url{www.shoup.net/ntl}} 
makes it practicable in numerous cases.

This work was partly motivated by an application from graph theory.
For a graph $X$ on $n$ vertices with vertex set $V(X)$ and edge set
$E(X)$ the $k$-th symmetric power $X^{k}$ is the graph with the
${n \choose k}$ $k$-subsets of $V(X)$ as vertices and with two such
$k$-subsets adjacent if their symmetric difference is in $E(X)$.

Graph theorists are interested in the spectrum of such graphs (defined as the
spectrum of their adjacency matrix) since they are closely related to the
description of their isomorphism class. 
More precisely, if a certain power $k$ is found, such that the symmetric $k$th 
power of a graph describes its isomorphism class, this would provide a polynomial
algorithm to solve the graph isomorphism problem. 

This motivated the team of Audenaert, Godsil, Royle and Rudolph to study in
\cite{Royle:2007:symm} the spectrum of  symmetric powers of a class of graphs:
the strongly regular graphs. They prove that there exist infinitely many graphs
having co-spectral symmetric squares. But concerning the symmetric cubes, no
pair of graph is known to have co-spectral symmetric cubes until now.

We helped Royle to investigate further the computation of the characteristic
polynomial of the symmetric cubes of strongly regular graphs. 
He was able to test the first 72 cases corresponding to the graphs with fewer
than 29 vertices. Using our implementations available in \texttt{LinBox}, we
have been able to test the $36\,582$  graphs with fewer than 36 vertices and
check that there is no pair of graphs among them having cospectral symmetric
cubes.

We used the matrices of this application to benchmark the implementations of the
previously presented algorithms.
These matrices are sparse and symmetric, and therefore especially suited to black-box
computations.
Moreover, several parameters such as the degrees of their
minimal polynomials or the average number $\omega$ of nonzero elements per row
vary among the matrices. The matrices \texttt{EX1, EX3, EX5} correspond
respectively to the symmetric cubes of the strongly regular graphs with
parameters (16,6,2,2), (26,10,3,4) and (35,16,6,8). Their dimensions are
respectively
$560 = \binom{16}{3}$, $2600 =
\binom{26}{3}$ and $6545 = \binom{35}{3}$.
The matrices \texttt{EX2} and \texttt{EX4} correspond to different graphs but
with similar parameters as \texttt{EX1} and \texttt{EX3}.

All the matrices used in the experiments, including adjacency matrices
of the symmetric powers, are avaible on-line in the {\em Sparse Integer
Matrices
Collection\footnote{\url{ljk.imag.fr/membres/Jean-Guillaume.Dumas/SIMC}}}.
In particular we used, in the following tables, matrices from the
\url{SPG}, \url{Forest}, \url{Trefethen} and  \url{Homology} sections of the
collection. When the tested matrix was not square, we considered the
square matrix obtained by padding it with zeroes.

\begin{table}[hptb]
\begin{center}
{\small
\begin{tabular}{lccccc}
\toprule
Matrix           &  EX1 & EX2 & EX3 & EX4 & EX5 \\ 
\midrule
$n$: dimension   & 560  & 560 & 2600 & 2600 & 6545 \\
$d$: deg ($P_\text{min})$ & 54  & 103 & 1036 & 1552 & 2874  \\
$\omega$: sparsity & 15.6 & 15.6 & 27.6 & 27.6 & 45.2\\
\midrule
$\Z$-\texttt{Minpoly} & 0.11s & 0.26s & 117s & 260s & 5002s  \\
\midrule
$\Z[X]$ factorize& 0.02s & 0.07s & 9.4 & 18.15 & 74.09s \\
\midrule
Nullity/comb.& \textbf{3.37s} & 5.33s & \textbf{33.2s} & \textbf{30.15s} &\textbf{289s} \\
Total& \textbf{3.51s} & 5.66s & \textbf{159.4s} & \textbf{308.1s} & \textbf{5366s}\\
\midrule
Index calc.  & 3.46s & \textbf{4.31s} & 64.0s & 57.0s & 647s\\
Total&  3.59s & \textbf{4.64s} & 190.4s & 336.4s & 5641s\\
\bottomrule
\end{tabular}
}
\end{center}
\caption{Computation time for tasks of the integer adaptive algorithm on a
  Pentium4 (3.2 GHz; 1 Gb)}\label{table:timings} 

\end{table}

We report in table \ref{table:timings} the computation time of the
different modules described in this paper. For each matrix, two computations are
compared: they share the computation of the minimal polynomial over $\Z$. Then
the determination of the multiplicities is done either by the combination of the
nullity algorithm and the combinatorial search (with the threshold $T$ set to 5), or by the index calculus method.

We first note that the determination of the multiplicities may be the dominant
operation when the degree of the minimal polynomial is small, as for the matrix
\texttt{EX1}. This makes the motivation for this study obvious.
For this task either method, nullity or resolution of the logarithmic system,
can be the most competitive option, depending on the structure of the irreducible
factors. This advocates for the adaptive approach of algorithm
\ref{alg:spcharpoly1} combining both methods, and the computation of the
$k$th invariant factor.

In order to emphasize the improvement of the black-box determination of the
multiplicities over dense methods, we now compare it to the alternative technique presented in
\cite[\S 4.2.2]{jgd:2005:charp}. This also relies on the computation of the
minimal polynomial in $\Z[X]$ and its decomposition into irreducible factors. But
the multiplicities are then obtained using one \emph{dense} computation of the characteristic
polynomial in a randomly chosen finite field. It is
therefore not anymore a black-box algorithm. 
We will denote it by \texttt{dchar}.  \texttt{Comb} is the nullity-combinatorial
search algorithm, and \texttt{ind} is the index calculus method. $A$ denotes \texttt{08blocks}, $B$ is \texttt{ch5-5.b3}, and $T$ is \texttt{Tref500} from the Sparse Integer Matrices Collection.

\begin{table}[htbp]
\begin{center}
\begin{tabular}{cccccc}
\toprule
Matrix & n & $\omega$ & \texttt{dchar} & \texttt{null-comb} &
\texttt{ind.}\\
\midrule
$A$ & 300 & 1.9  & 0.32s & 0.08s & \textbf{0.07s}\\
$AA^T$          & 300 & 2.95 & 0.81s & \textbf{0.12s} & \textbf{0.12s}\\
$B$   & 600 & 4    & 4.4s  & \textbf{1.52s} & 1.97s\\
$BB^T$          & 600 & 13   & \textbf{2.15s}  & 3.96  &7.48s \\
\texttt{TF12}         & 552 & 7.6  & 6.8s   & \textbf{5.53s} & 5.75s\\
\texttt{mk9b3}        & 1260 & 3   & 31.25s & \textbf{10.51s} & 177s \\ 
\texttt{Tref500}      & 500 & 16.9 & 65.14s & \textbf{25.14s} & 25.17s\\
\bottomrule
\end{tabular}
\end{center}
\caption{Integer black-box approach for 
  multiplicities on an Athlon (1.8 GHz; 2 Gb)}
\label{tab:ciavsbn}
\end{table}

Table \ref{tab:ciavsbn} shows the improvement of the black-box approach for
several matrices coming from
different applications. Once again, the structure of the irreducible factors of
the minimal polynomials cause various behaviors for each variant. For example
the times of 
\dlogsys\xspace are similar to those of 
\nullcomb, sometimes better but also sometimes much slower, as
for the matrices $BB^T$ and \texttt{mk9b3}.

\section{Conclusion}
We developed several ways to recover the multiplicities of the factors of the 
characteristic polynomial  from a factorization of the minimal
polynomial. 
Over a finite field  hybrid heuristics are proposed, that compete with
the best theoretical complexity.
Over the ring of integers, our approach enables fast
computations particularly when the coefficients or degree of the minimal
polynomial are small.
This is illustrated on a family of strongly regular graphs, in order to verify
 that there are no symmetric co-spectral cubes.

Further studies on the theoretical complexity remain to be done, and could lead
to better implementations in practice. In particular, a recent algorithm for
dense matrices \cite{Pernet:2007:charp} might be adapted for black-box
matrices. 
In this regard, extending the block projections of
\cite{EGGSV:2007:BBinv} to the case of similarity transformations would play a
crucial role.


\end{document}